\documentclass[runningheads,a4paper]{llncs}

\usepackage{amsmath,amssymb,amsfonts}
\usepackage{graphicx}

\usepackage{epsf}
\usepackage{tabularx,ragged2e,booktabs}

\usepackage{tikz}
\usepackage{url}
\usepackage{mathtools}

\usetikzlibrary{arrows,automata}
\usepackage[ruled,vlined,linesnumbered]{algorithm2e}

\usepackage{times}
\usepackage{array}
\usepackage{multirow}

\newcommand{\reals}{\mathbb{R}}
\newcommand{\naturals}{\mathbb{N}}

\newcommand{\MM}{\mathcal{M}}
\newcommand{\PP}{\mathbf{P}}
\newcommand{\UU}{\mathbf{U}}
\newcommand{\HC}{\textsc{HyChecker}}
\newcommand{\Beta}{\mathrm{B}}

\makeatletter
\newcommand{\distas}[1]{\mathbin{\overset{#1}{\kern\z@\sim}}}%
\newsavebox{\mybox}\newsavebox{\mysim}
\newcommand{\distras}[1]{%
	\savebox{\mybox}{\hbox{\kern3pt$\scriptstyle#1$\kern3pt}}%
	\savebox{\mysim}{\hbox{$\sim$}}%
	\mathbin{\overset{#1}{\kern\z@\resizebox{\wd\mybox}{\ht\mysim}{$\sim$}}}%
}
\makeatother

\newcommand\numberthis{\addtocounter{equation}{1}\tag{\theequation}}

\DeclareMathOperator*{\argmin}{arg\,min}

\usepackage{url}
\urldef{\mailsa}\path|{pingfan_kong, xiaohong_chen, sunjun, jingyi_wang}@sutd.edu.sg|
	\urldef{\mailsb}\path|sunmeng@math.pku.edu.cn|
	\urldef{\mailsc}\path|liyi_math@pku.edu.cn|

\begin{document}
	
	\mainmatter  % start of an individual contribution
	
	\title{Towards Concolic Testing for Hybrid Systems}
	
	\titlerunning{Lecture Notes in Computer Science}
	
	\author{
		Pingfan Kong$^1$ \and Yi Li$^2$ \and Xiaohong Chen$^{1,3}$ \and
		Jun Sun$^1$ \and Meng Sun$^2$ \and Jingyi Wang$^1$}

	\institute{$^1$Singapore University of Technology and Design, Singapore~~~~~$^2$Peking University, China\\
		$^3$University of Illinois at Urbana-Champaign, United States\\
%		\mailsa\\
%		\mailsb\\
%		\mailsc\\
		}
	
	
	\toctitle{Analyzing Hybrid Systems with Concolic Sampling}
	\tocauthor{Authors' Instructions}
	\maketitle

% Jul 5th, Pingfan: Why we should use importance sampling? Will it work if we merely increase sample rate throughout (0,1) for covering the rare mode? We proved sufficiency, but not necessity
\begin{abstract}
Hybrid systems exhibit both continuous and discrete behavior. Analyzing hybrid systems is known to be hard. % as they can both flow (described by differential equations) and jump (described by program-like statements). 
Inspired by the idea of concolic testing (of programs), we investigate whether we can combine random sampling and symbolic execution in order to effectively verify hybrid systems. We identify a sufficient condition under which such a combination is more effective than random sampling. Furthermore, we analyze different strategies of combining random sampling and symbolic execution and propose an algorithm which allows us to dynamically switch between them so as to reduce the overall cost. Our method has been implemented as a web-based checker named \textsc{HyChecker}. \textsc{HyChecker} has been evaluated with benchmark hybrid systems and a water treatment system in order to test its effectiveness. %Our experiment results show that \textsc{HyChecker} is more effective than existing methods in falsifying properties or proving their absence.
%One obstacle is to decide the time points when mode jumps happen. This needs solving the ordinary differential equations with specific guard conditions which is time-consuming and in most occasions closed forms are not available. In previous work, a method for analyzing the dynamics of a hybrid system $H$ in terms of a Markov chain $M$ was proposed. Our new algorithm is based on this method. We have two major improvements. One is that we parallelized the algorithm to fully utilize the multicore processors by generating trajectories starting from multiple states at the same time. The other is that we borrow the idea of concolic testing to find if there exists any mode whose guard is rarely satisfied, and thus not efficient with random testing, and call dReach to find an importance region satisfying this guard. Then we develop a new importance sampling mechanism to improve our efficiency of random testing. We will brief this confirmation report with the backgrounds as well as our formal proves and our proposed new algorithm. A simple test run is also given in the experiments chapter.
\end{abstract}

%A category including the fourth, optional field follows...
%\category{D.2.8}{Software Engineering}{Metrics}[complexity measures, performance measures]


%\keywords{Hybrid systems, symbolic execution, random testing}

\section{Introduction} \label{intro}
Hybrid systems are ever more relevant these days with the rapid development of the so-called cyber-physical systems and Internet of Things. Like traditional software, hybrid systems rely on carefully crafted software to operate correctly. Unlike traditional software, the control software in hybrid systems must interact with a continuous environment through sensing and actuating. Analyzing hybrid systems automatically is highly nontrivial. %The control logic implemented in the control program reacts to stimulus from the environment, which in turn is influenced by the actuating from the previous steps. Thus, to analyze a hybrid system, we must first develop a model of the environment that the control software is interacting with and then analyze the system as a whole. %Modeling the environment is complicated and the model is often composed of ordinary differential equations (ODE).
With a reasonably precise model of the entire system (e.g., in the form of a hybrid automaton), analyzing its behaviors (e.g., answering the question whether the system would satisfy a safety property) is challenging due to multiple reasons. Firstly, the dynamics of the environment, often composed of ordinary differential equations (ODE), is hard to reason about. For instance, there may not be closed form mathematical solutions for certain ODE. Secondly, unlike in the setting of traditional model checking problems, the variables in the hybrid models are often of real type and the (mode) transitions are often guarded with propositional formulas over real variables. There have been theoretical studies on the complexity of analyzing hybrid systems. For instance, it has been proved that non-trivial verification and control problems on non-trivial nonlinear hybrid systems are undecidable~\cite{DBLP:conf/lics/Henzinger96,DBLP:journals/jcss/HenzingerKPV98}. As a result, researchers have proposed to either work on approximate models of hybrid systems~\cite{DBLP:conf/tacas/HenzingerM00,DBLP:journals/fmsd/HahnHHK13}, or adopt approximate methods and tools on the hybrid system models~\cite{benjaminsmc,barbot2012importance,barbot2012coupling}.

One line of research (which we believe is relevant) is on analyzing the behaviors of hybrid systems through \emph{controlled} sampling. One example of those sampling-based methods is~\cite{benjaminsmc}. The idea is to approximate the behavior of a hybrid system probabilistically in the form of \emph{discrete time Markov chains} (DTMC). The complex dynamics in hybrid automata model is approximated using numeric differential equations solvers, and the mode transitions are approximated by probabilistic transition distributions in Markov chains. %In~\cite{benjaminsmc}, a Markov chain is constructed based on the assumption that the probability of a mode transition is proportional to the measure of the value state and time point pairs at which this transition is enabled. %Algorithmically, the Markov chain is constructed on-the-fly using the \emph{Monte Carlo} simulation. That is, we estimate the probability of a mode transition by obtaining many random samples and approximate the probability by measuring how often a transition guard is satisfied.
Afterwards, methods like hypothesis testing can be applied to the Markov chain to verify, probabilistically, properties against the original hybrid model. %For instance, in~\cite{benjaminsmc}, the authors showed that an equivalence relation between the hybrid automaton model and the Markov chain can be established such that any safety property expressed as a bounded linear temporal logic (BLTL) formula of a hybrid system is preserved in the corresponding Markov chain, and vice versa~\cite{benjaminsmc}.

While sampling-based methods like~\cite{benjaminsmc} are typically more scalable, there are limitations. Arguably, the most important one is that random sampling does not work well when the system contains \emph{rare events}, i.e., events which by definition are unlikely to occur through random sampling. When systems get complicated, every event becomes rarer in a way. Existing remedies for this problem include importance sampling~\cite{barbot2012importance} and importance splitting~\cite{DBLP:conf/cav/JegourelLS13}, which work by essentially increasing the probability of the rare events. Both approaches are however useful only in certain limited circumstances.% (refer to~\cite{barbot2012importance,DBLP:conf/cav/JegourelLS13} for details).
% Firstly, sampling-based methods do not prove a property. One would argue that sampling-based methods (like program testing) is designed for falsifying a property. However, if a counterexample is never sampled, the result is often a probabilistic guarantee on the property (e.g., the property is probabilistically true with certain confidence), which is often hard to interpret. Secondly, it is known that random sampling does not work well when the hybrid systems contains \emph{rare events}.

One potential remedy for the problem is concolic testing, which is a technique proposed for analyzing programs~\cite{Godefroid:2005:DDA:1064978.1065036,sen2007concolic}.
%In this work, we propose an alternative approach which aims to retain the benefit of sampling-based methods whilst tackling the problem of rare events. It is inspired by the method called concolic testing (a.k.a.~dynamic symbolic execution) proposed for analyzing programs~\cite{Godefroid:2005:DDA:1064978.1065036,sen2007concolic}.
The idea is: if random sampling fails to fire certain transitions in certain state (i.e., a potential rare event), we apply symbolic execution to generate the specific inputs which would trigger the transition or to show that the transition is infeasible. In this work, we investigate the possibility of applying concolic testing to hybrid systems. In particular, we study two fundamental questions. One is under what condition combining random sampling and symbolic execution is beneficial, i.e., given a property, is it guaranteed to find a counterexample with a smaller number of samples? The other is, among different strategies of combining random sampling and symbolic execution (i.e., when and how to apply symbolic execution), how do we define and identify the more effective strategies?
%considering random sampling is often more efficient than symbolic execution, how do we combine them so that it is the most cost-effective?
We remark that the latter question is particularly relevant to the analysis of hybrid systems as symbolic execution for hybrid automata is often very time consuming and thus a good strategy should perhaps be: applying symbolic execution as minimum as possible. Based on the answers, we then design an algorithm which adopts a strategy to dynamically switch between random sampling and symbolic execution. Intuitively, it works by continuously estimating whether certain transition is rare or not and applying symbolic execution only if the transition is estimated to be rarer than certain threshold. Furthermore, the threshold is calculated according to a cost model which estimates the cost of symbolic execution using certain constraint solver. Our method has been implemented as a self-contained web-based checker named \textsc{HyChecker} and evaluated with benchmark hybrid systems as well as a water treatment system in order to test its effectiveness. %ur experiment results show that \textsc{HyChecker} is more effective than existing methods in falsifying properties.

%In view of the complexity of solving constraints obtained from the hybrid model, our approach is tailored towards finding a local infeasibility proof, incrementally expanding its searching scope for the proof. While the above-proposed concolic testing would solve the problem (i.e., the solver would discover that a transition is feasible by generating concrete settings where the transition would be enabled at the state), we would like to retain the advantage of sampling and limit the usage of constraint solving. Thus, we further improve our method with importance sampling. That is, after a constraint solver discovers a rare event, we would use the constraint solving results to increase locally the probabilistic weight of the rare event so that more traces containing the rare events are generated so that more behaviors after the rare events will be sampled. We refer the readers to Section~\ref{sec:relatedworks} for a review of the related work. Our method has been implemented as a self-contained web-based checker named \textsc{HyChecker}, which is available online~\cite{hychecker}. \textsc{HyChecker} has been evaluated with benchmark hybrid systems as well as a real-world complicated water treatment system in order to test its effectiveness. Our experiment results show that \textsc{HyChecker} is more effective than existing methods in falsifying properties or proving their absence.

The remainders of the paper are organized as follows. In Section~\ref{sec:HA&DTMC}, we define a DTMC interpretation of hybrid system models. In Section~\ref{approach}, we view symbolic execution as a form of importance sampling and establish a sufficient condition for importance sampling to be beneficial. In Section~\ref{sec:strategy}, we discuss strategies on combining random sampling and symbolic execution. In Section~\ref{experiments}, we present our implementation and evaluate its effectiveness. In Section \ref{sec:relatedworks}, we conclude and review related work. %Due to space limit, a detailed example on how our approach works is shown in Appendix A.

%importance sampling, concolic testing in program verification and gave introduction to the tool dReach \cite{kong2015dreach}. In Section \ref{sec:implementation}, we first introduce how to build a Markov Chain based on the original hybrid automata input and why through studying this Markov Chain we can equally prove the properties specified in the hybrid automata. We then talked about how to adopt dReach to return an importance region. We then introduce our main contribution of a new importance sampling mechanism based on the generated importance region to find rarely touched modes. Then we give our general algorithm for this method and introduce our tool \emph{SMChecker}. In Sec \ref{experiments} we wrote a simple hybrid system with a rarely satisfied mode in it. We show in our result how our tool reduces time consumption compared to general random testing. In Section \ref{sec:relatedworks}, we will talk into some related works for probabilistic model checking tools, symbolic verification methods and some other mechanisms dealing with inefficient sampling. In Sec \ref{futureworks}, we tell our direction for future works.

\section{A Probabilistic View}\label{sec:HA&DTMC}
In this section, we present a probabilistic interpretation of hybrid system models, which provides the foundation for defining and comparing the effectiveness of random sampling, symbolic execution or their combinations.
%We first give a primer on hybrid automata~\cite{henzinger2000theory}, which is often used to model hybrid systems. We then introduce discrete-time Markov chains (DTMC) as an abstraction of hybrid automata. The approach shown in this section is a formalization of~\cite{benjaminsmc}. \\
In this work, we assume that hybrid systems are modelled as hybrid automata~\cite{henzinger2000theory}. The basic idea of hybrid automata is to model different discrete states in a hybrid system as different \emph{modes} and use differential equations to describe how  variables in the system evolve through time in the modes. For simplicity, we assume the differential equations are in the form of ordinary differential equations~(ODEs).
\begin{definition} A \emph{hybrid automaton} is a tuple ${\cal H} = (Q,V,q_0,I,\{f_q\}_{q\in Q},\{g_{(q,p)}\}_{\{q,p\}\subseteq Q})$ such that $Q$ is a finite set of modes; $V$ is a finite set of state variables; $q_0 \in Q$ is the initial mode; $I \subseteq \reals^n$ is a set of initial values of the state variables; $f_q$ for any $q \in Q$ is an ODE describing how variables in $V$ evolve through time at mode $q$; and $g_{(q,p)}$ for any $q,p\in Q$ is a guard condition on transiting from mode $q$ to mode $p$.
\end{definition}
For simplicity, we often write $q\xrightarrow{g}p$ to denote $g_{(q,p)}$. For example, the hybrid automaton shown on the left of Figure~\ref{fig:osc_automata} models an underdamped oscillatory system~\cite{gordon2007}, such as a spring-mass or a simple pendulum with a detector that raises an alarm whenever the displacement $x$ exceeds the threshold $a$. %The behavior is captured by the following differential equation
%\begin{equation}
%\frac{\mathrm{d}^2 x}{\mathrm{d}t^2} + \frac{\mathrm{d}x}{\mathrm{d}t} + \omega^2 x = 0
%\label{eq:osc}
%\end{equation}
%with $\omega = 2\pi$.
The initial displacement $x(0)=0$, while its tendency to deviate from the equilibrium $x'(0) \in [0,2\pi]$. An alarm is raised when the system enters mode $q_e$, which is reachable only through the transition $q_0 \xrightarrow{x>a} q_e$.
%We are interested in whether the detector will alarm in the first five seconds, i.e., whether there exists a solution $x(t)$ of the ODE at mode $q_0$ with $x(0)=0$ and $x'(0)\in[0,2\pi]$ such that $x(t) > a$.

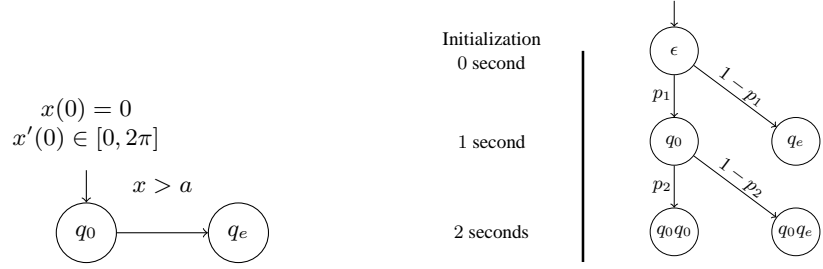
\begin{figure}[t]
	\scalebox{1.0}{
		\begin{tikzpicture}[->, node distance = 2cm,scale = 1]
		\tikzstyle{every node} = [align = center,minimum size = 1.2 cm];
		\node [initial by arrow,initial text={$x(0)=0$\\$x'(0)\in[0,2\pi]$},initial where = above,state,] (INIT){$q_0$};
		\node [state](T1)[right of = INIT]{$q_e$};
		
		\path
		(INIT)edge node [above=0.01cm] {$x>a$}(T1);
		
		\end{tikzpicture}
	}
~~~~~~~~~~~~~~~~~~~~~~~~
	\scalebox{0.8}{
		\begin{tikzpicture}
		%\draw[help lines] (-2,0) grid (3,-10);
		
    	\draw[very thick] (-1.5,0) -- (-1.5,-3.5);
    	
    	\draw (-3,0) node [align = center]
    	    {Initialization\\
    	    	$0$ second};
    	\draw (-3,-1.5) node {$1$ second};
    	\draw (-3,-3) node {$2$ seconds};
	
		\tikzstyle{every node} = [inner sep = 1pt];
		\node [initial by arrow,state,initial where = above, initial text = {}] at (0,0) (INIT){$\epsilon$};
		\node [state] at (0,-1.5) (T1)  {$q_0$};
		\node [state] at (2,-1.5)  (T1e)   {$q_e$};
		\node [state] at (0, -3) (T2)   {$q_0q_0$};
		\node [state] at (2,-3)  (T2e) {$q_0q_e$};
		
		\path[->]
		(INIT)edge node [left] {$p_1$} (T1)
		(INIT) edge node [above, rotate=-36.9] {$1-p_1$} (T1e)
		(T1) edge node [left] {$p_2$} (T2)
		(T1) edge node [above, rotate=-36.9] {$1-p_2$} (T2e)
		;

		\end{tikzpicture}
	}
	\caption{An oscillatory system: the ODE at $q_0$ is $x''+x'+4\pi^2 x = 0$ and the one at $q_e$ is $x'=0$} %with a displacement detector
	\label{fig:osc_automata}
\end{figure}

Next, we define the semantics of a hybrid automaton in the form of an infinite-state labeled transition system (LTS).
\begin{definition} Let ${\cal H} = (Q,V,q_0,I,\{f_q\}_{q\in Q},\{g_{(p,q)}\}_{\{p,q\}\subseteq Q})$ be a hybrid automaton. The semantics of ${\cal H}$, written as $sem({\cal H})$, is an LTS $(S, S_0, T, \rightarrow)$, where $S = \{(q,v) \mid q \in Q \text{ and } v\colon V \to \reals^n\}$ is the set of all (concrete) states; $S_0 = \{(q_0,v) \mid v \in I\}$ is the set of initial states; $T = \reals_+ \cup \{\epsilon\}$ is the set of transition labels, where $\epsilon$ is a label for all discrete jumps; and $\rightarrow \subseteq S \times T \times S$ contains two sets of transitions. One is time transitions, i.e., $(q,v)\xrightarrow{t}(q,u)$ if there exists a solution $\xi$ to the differential equation $\mathrm{d}V/\mathrm{d}t=f_q(V)$ such that $\xi(0)=v$ and $\xi(t)=u$. The other is jumps, i.e., $(q,v)\xrightarrow{\epsilon}(p,v)$ where there exists a transition $q\xrightarrow{g}p$ in $\cal H$ such that $v$ satisfies $g$.
\end{definition}
A finite \emph{run} $\rho$ of $\cal H$ is a finite sequence of transitions of $sem({\cal H})$.
%\[\rho=(q_0,v_0)\xrightarrow{t_1}(q_1,v_1)\xrightarrow{t_2}\cdots\xrightarrow{t_k}(q_k,v_k)\]
%where $k$ is the \emph{length} of the run $\rho$. In this work,
Since we are investigating random sampling and symbolic execution (both of which are limited to finite runs), we focus on runs of \emph{bounded} length. For simplicity, we assume that all finite runs can be extended to an infinite non-Zeno run (such that time elapses unboundedly~\cite{henzinger2000theory}). It is straightforward to see that there always exists a time unit $\Delta t > 0$ such that at most one jump (i.e., $\epsilon$-transition) occurs with $\Delta t$ time units. In the following, we simply assume that $\Delta t$ is defined as one time unit. %In other words, the dynamic behavior of hybrid automata is a step-by-step one in terms of mode transitions.
As a result, by observing the system mode at the end of every time unit, we can obtain a \emph{trace} of $\cal H$ as $\pi =  q_0q_1\dots q_k $, i.e., the sequence of modes observed during the run. We remark that if there is no jump during the time unit, the same mode is observed.

In the following, we focus on reachability analysis of certain modes~\cite{benjaminsmc}, i.e., certain modes in $\mathcal{H}$ are considered negative and we would like to check if any of them is reachable.
%We focus on analyzing properties in the form of bounded linear temporal logic (BLTL) formulae~\cite{DBLP:journals/apal/Kamide12}. %A BLTL formula is defined by the following grammar:
%\begin{equation*}
%\varphi \Coloneqq true \mid a \mid \varphi_1 \land \varphi_2 \mid \lnot \varphi \mid \Circle \varphi \mid \square^{\le t} \varphi \mid \lozenge^{\le t} \varphi \mid \varphi_1 \mathbin{\mathsf{U}}^{\le t} \varphi_2
%\end{equation*}
%where $a$ is an atomic proposition that we concern. A BLTL formula induces a partition $\Pi = \Pi^+ \cup \Pi^-$ of positive traces that satisfy the formula and negative traces that violate the formula.
%We skip details on BLTL definition and semantics, and remark that the problem of verifying a BLTL formula constituted by propositions on the system mode can be reduced to a reachability analysis of certain modes~\cite{benjaminsmc}.
%Thus, we focus on reachability of some negative modes in $\mathcal{H}$ in the following.
We remark that the verification problem of properties expressed in BLTL formula~\cite{DBLP:journals/apal/Kamide12} can be reduced to reachability analysis~\cite{benjaminsmc}. For instance, in the example shown in Figure~\ref{fig:osc_automata}, the safety property is reduced to whether the negative mode $q_e$ is reachable or not (with certain time). A trace is positive if it contains no negative mode. It is negative (a.k.a.~a counterexample) if it contains at least one negative mode. %Negative traces are also referred to as \emph{counterexamples}.
% \\

%\noindent \emph{Example} In the oscillatory example, a concrete state is of the form $(x',x)^T$. A possible run is
%$$\langle q_0, (6,0^T)\rangle \xrightarrow{0.2} \langle q_0, (1.2866,0.8233)^T\rangle \xrightarrow{\epsilon} \langle q_e, (0,0.8233)^T\rangle.$$
%The corresponding trace is $q_0q_0q_e$. Assume that the only negative mode is $q_e$. The trace is therefore considered as a counterexample. \\
%
%%xiaohong
%%WE NEED TO CHANGE ``STATES'' OF MARKOV CHAIN TO ``NODES''

%\noindent \textbf{Markov chain abstraction} In order to answer the question on how many samples are required to identify a counterexample,
Next, we introduce a Markov chain interpretation of $\cal H$, adopted from~\cite{benjaminsmc}. Without loss of generality, we assume a uniform probability distribution on all initial states. This uniform distribution naturally induces a probability distribution over the traces of the system. %Our work can be extended to the scenario where a different distribution for sampling is used.
Recall that a transition $q\xrightarrow{g}p$ of $\cal H$ can be fired only when $g$ is satisfied. Suppose the system is in the state $(q,v)$ initially and becomes $(q,v_t)$ after a time transition of $t$. If $v_t$ satisfies $g$,  the transition is enabled. %If $v_t$ does not satisfy $g$, then such a transition to $p$ cannot take place at time $t$.
We denote the set of all points in time within $(0,1)$ when the mode transition $q\xrightarrow{g}p$ is enabled as
\begin{equation}
\label{eq:timewindow}
T_q(v,g) = \{t \in (0,1) \mid \theta_q(v,t) \text{ satisfies } g\},
\end{equation}
where $\theta_q(v,\cdot)$ is the solution of the ODE at mode $q$ with the initial value $v$.
If the transition is fired at some time point $t \in T_q(v,g)$, the following state is observed after 1 time unit: $(p,\theta_p(v_t,1-t))=(p,\theta_p(\theta_q(v,t),1-t))$.
That is, the new mode is $p$ and the variables evolve according to the ODE at mode $q$ for $t$ time unit and then according to the ODE at mode $p$ for $1-t$ time unit.
For simplicity, we write $v_{q,p}(v,t)$ to denote the variable state reached from state $(q,v)$ by firing transition $q\xrightarrow{g}p$ at time $t$, i.e., $v_{q,p}(v,t) = \theta_p(\theta_q(v,t),1-t)$. Furthermore, we write $ v_{q,p}(v)$ to denote the set of all variable states reached from state $(q,v)$ by firing transition $q\xrightarrow{g}p$ at any time the transition is enabled, i.e., $v_{q,p}(v) = \{v_{q,p}(v,t)\mid t\in T_q(v,g)\}$.

By our assumption on the uniform random sampling,
%Now, assume we only know that the automaton $\cal H$ transits from mode $q$ to $p$ within the time window $(0,\Delta t)$, but we are not aware of the exact time in which the transition happens, and as a result we are not aware of the current state of $\cal H$. We merely know for sure that the transition must take place at a point $t \in T_q(v,g)$ in time. Without further information about the transition, we cannot cast any preference for any specific time point $t \in T_q(v,g)$, so we have to assume that the exact time when the transition happens is equally likely to be any time in $T_q(v,g)$, i.e.,
there is a uniform distribution over $T_q(v,g)$. This uniform distribution, denoted as $\UU(T_q(v,g))$, naturally induces the following probability distribution over $v_{q,p}(v)$
\begin{equation}\label{eq:py}
\PP(Y) = \int_{t \in T_q(v,g)}\frac{[v_{q,p}(v,t)\in Y]}{\Vert T_q(v,g) \Vert}\mathrm{d}t
\end{equation}
for any $Y \subseteq v_{q,p}(v)$, where $[\  \cdot\ ]$ is the Iverson bracket~\cite{Iverson1962} and $\Vert \cdot \Vert$ is the Lebesgue measure~\cite{Lebesgue1902}. Intuitively, if initially the system is in the state $(q,v)$, we obtain a probability distribution over all possible states after taking the transition.

Next, we generalize the result so as to compare the probability of taking different transitions from different initial states. % in two perspectives. Firstly, we will allow the automaton to make any enabled transitions from $q$ instead a specific one. Secondly, we will let the automaton $\cal H$ start from a collection of states $(q,X)$ where $X \subseteq \reals^n$. The generalization seems to be tedious but are in fact standard and straightforward.
We assume a probability space $(X,\PP)$ where $X \subseteq \reals^n$, and the automaton $\cal H$ starts from a state $(q,v)$ with $v \sim \PP$. Let $q \xrightarrow{g_i} p_i$ where $i \in \{1,\dots,m\}$ be the transitions from $q$. Given an initial state $(q,v)$, the time window in which the transition to $p_i$ is enabled is $T_q(v,g_i)$. We assume that the system does not favor certain transitions and the probability of taking a transition is proportional to the size of the time window in which that transition is enabled. In other words, the probability of taking the transition $q\xrightarrow{g_i}p_i$ from state $(q,v)$ is defined as $p_{q,p_i}(v) = {\Vert T_q(v,g_i) \Vert}/{\sum_{j=1}^{m}\Vert T_q(v,g_j) \Vert}$.
According to the law of total probability, we have
\begin{equation}\label{eq:transdist}
p_{q,p_i} = \int_{v\in X}\frac{\Vert T_q(v,g_i) \Vert}{\sum_{j=1}^{m}\Vert T_q(v,g_j) \Vert} \mathrm{d}\PP(v).
\end{equation}
%%%sunjun: add intuitive explanation on what this defines here.
Furthermore, assume the transition $q \xrightarrow{g_i}p_i$ is fired. Given the condition that $v$ is a fixed $v_0$, we know the conditional probability distribution over $v_{q,p}(v_0)$, for any $Y \subseteq v_{q,p}(v_0)$, is defined as:
\begin{equation*}
\PP(Y\mid v=v_0) = \int_{t \in T_q(v_0,g)}\frac{[v_{q,p}(v_0,t)\in Y]}{\Vert T_q(v_0,g) \Vert}\mathrm{d}t.
\end{equation*}
%%%sunjun: add intuitive explanation on what this defines here.
By the law of total probability, for any $Y \subseteq v_{q,p}(X)$, we have
\begin{align*}
\PP(Y) &= \int_{v\in X}\PP(Y \mid v) \mathrm{d}\PP(v)\\
&= \int_{v\in X}\int_{t \in T_q(v,g)}\frac{[v_{q,p}(v,t)\in Y]}{\Vert T_q(v,g) \Vert}\mathrm{d}t\mathrm{d}\PP(v)\\
&= \int_{v\in X}\int_{0}^{1}\frac{[t \in T_q(v,g) \land v_{q,p}(v,t)\in Y]}{\Vert T_q(v,g) \Vert}\mathrm{d}t\mathrm{d}\PP(v). \numberthis \label{eq:statedist}
\end{align*}
%%%sunjun: add intuitive explanation on what this defines here.
Equations~\eqref{eq:transdist} and \eqref{eq:statedist} effectively identify a discrete-time Markov chain~(DTMC).

\begin{definition}\label{def:MC}
Let ${\cal H} = (Q,V,q_0,I,\{f_q\}_{q\in Q},\{g_{(p,q)}\}_{\{p,q\}\subseteq Q})$ be a hybrid automaton, and $K$ be a bound of trace length. The DTMC associated with $\cal H$ is a tuple ${\cal M}_{\cal H} = (S, u_0, Pr)$ where a node in $S$ is of the form $(q, X, \PP_X)$ where $q \in Q$ is a mode, $X$ is the set of values for $V$ and $\PP_X$ is a probability distribution of the values in $X$; the root $u_0 = (q_0, I, \UU_I)$ where $\UU_I$ is the uniform distribution over $I$; and for any $(q, X, \PP_X) \in S$, the transition probability $Pr((q, X, \PP_X), (p, v_{q,p_i}(X), \PP_i))$, where the probability distribution $\PP_i$ is defined as in equation~\eqref{eq:statedist}.
% is a weighted tree with a depth of $K$ and a degree of $\Vert Q \Vert$. The Markov chain is defined inductively as follows.
%	\begin{itemize}
%		\item The root of $\cal M$ is $(q_0, I, \UU_I)$ where $\UU_I$ is the uniform distribution on $I$.
%		\item For any node $(\pi, X, \PP_X)$ in $\cal M$ where the trace $\pi$ ends in state $q$, and for any mode transition $q \xrightarrow{g_i} p_i$ in $\cal H$, there is a node $(\pi p_i, v_{q,p_i}(X), \PP_i)$ in $\cal M$ where the probability distribution $\PP_i$ is defined as in equation~\eqref{eq:statedist}. The transition probability from $(\pi, X, \PP_X)$ to $(\pi p_i, v_{q,p_i}(X), \PP_i)$ is $p_{q,p_i}$ as defined in equation~\eqref{eq:transdist}.
%	\end{itemize}
\end{definition}
We remark that ${\cal M}_{\cal H}$ abstracts away the complicated ODE in $\cal H$ and replaces the guarded transitions with probabilistic transitions. A \emph{path} of ${\cal M}_{\cal H}$ with non-zero probability always corresponds to a trace of $\cal H$~\cite{benjaminsmc}. The partition of positive and negative traces in $\cal H$ naturally induces a partition of positive and negative paths in ${\cal M}_{\cal H}$. Notice that ${\cal M}_{\cal H}$ is by construction in the form of a tree. The degree of the tree is bounded by the number of modes in $\cal H$, and its depth is bounded by $K$, i.e., the bound on trace length.

For instance, following the above discussion, we can construct the DTMC of the oscillatory system shown on the right of Figure~\ref{fig:osc_automata}. The root node is $s_0 = (q_0,I,\UU_I)$ where $I = \{0\} \times [0,2\pi]$ and $\UU_I$ is the uniform distribution over $I$. There is one outgoing transition $q_0 \to q_e$ at mode $q_0$. Thus $s_0$ has two children nodes $s_1$ and $s_2$, where $s_1$ represents automaton taking the transition $q_0\to q_e$ in the first second, and $s_2$ represents automaton staying in mode $q_0$. For this simple example, we can analytically compute the transition probability, e.g., $p_1$ and $p_2$ shown in Figure~\ref{fig:osc_automata}. In general it is difficult. %We show in the next section how to approximate the transition probability through sampling. %The DTMC of the oscillatory system is shown in Figure~\ref{fig:DTMCexample}, where $p_1,\cdots ,p_5$ defined as in Definition~\ref{def:MC}. %We will return to this model to show how to estimate $p_1,\cdots,p_5$ using Bayesian inference in Section~\ref{sec:algorithms}.

\section{Symbolic Execution as a Form of Importance Sampling}\label{approach}
In this section, we analyze the effectiveness of random sampling and symbolic execution based on the DTMC interpretation of $\cal{H}$ developed in the previous section. In particular, we review symbolic execution as a form of importance sampling~\cite{Veach1995}, which intuitively speaking alters the probability distribution of ${\cal M}_{\cal H}$ in certain way so that a negative path is more likely to be sampled. In the following, we first define a way of measuring the effectiveness of random sampling, symbolic execution and possibly other sampling methods.

\subsection{Bayesian Inference}\label{sec:beyinf}
%Hybrid automata are not probabilistic models. By ``stepwising'' a hybrid automaton, we get the associated Markov chain as a probabilistic interpretation of the automaton in terms of mode transitions. This probabilistic point of view flourishes our theoretical results and provides a profound explanation on our experimental results. In this section we will systematically introduce our approach.
Recall that traces of $\cal H$ are partitioned into either positive trace, denoted as $\Pi^+$, or negative traces, denoted as $\Pi^-$. The probability of the system exhibiting a negative trace is called the \emph{error probability} and is denoted as $\theta = \PP(\Pi^-)$. Intuitively, after observing some sample traces (obtained either through random sampling or symbolic execution), we gain certain information on $\theta$.
%Furthermore, many transitions in $\cal M$ are labeled with the probability of zero.
%Recall that ${\cal M}_{\cal H}$ induces a probability space $(\Pi, \PP)$ on the set of all traces of $\cal H$. Given a partition $\Pi = \Pi^+ \cup \Pi^-$ of positive and negative traces, $\PP(\Pi^-)$ is the probability that a negative mode is visited. We call $\PP(\Pi^-)$ the \emph{error probability} and denote it as $\theta$.
Formally, we investigate the following questions: (1) how do we claim that the error probability $\theta$ is bounded by certain tolerance level $\delta$ and (2) how do we measure the confidence of the claim?

%Our approach aims to show that $\theta$ is controlled and bounded by the preset level $\delta$. In other words, we aim to show that the systems-under-test is \emph{robust} (rmk: whether ``robust'' is a good word?).
%As mentioned in the previous section, in this paper we are interested in whether or not the error probability $\theta$ is bounded by a preset tolerance level $\delta$.
We answer the questions based on statistical inference. Intuitively,
%, since both random sampling and symbolic execution can be seen as a form of sampling,
if we see many negative trace samples, we conclude \emph{with certain confidence} that the system is likely to have an error probability that is larger than the tolerance level $\delta$. If we identify few or even no negative traces, we conclude \emph{with certain confidence} that the system is likely to have an error probability within the tolerance level $\delta$. %In the following we will formalize the above intuition and show that 1) when we can claim that the error probability $\theta$ is less than the tolerance, and 2) how much confidence we have about our claim.
Formally, let random variable $X$ denote whether a sample trace is positive or negative, i.e., $\PP(X = 1)$ is the error probability $\theta$. Let $B(N,\theta)$ denote the binomial distribution with parameters $N \in \naturals$ and $\theta \in [0,1]$. We have $X \sim B(1,\theta)$. Given $N$ independent and identically distributed sample traces, the number of negative traces is: $m = X_1 + X_2 + \dots X_N \sim B(N,\theta)$.
Initially, before witnessing any sample trace, we may only estimate the value of $\theta$ based on historical data. %For example, $\theta$ is unlikely 0, since software usually contains bugs. Also, it is highly possible that $\theta < 1$; otherwise the system is full of bugs.
We thus assume a \emph{prior knowledge} of $\theta$ in the form of a \emph{prior distribution} $f(\theta)$. If no historical data are available, we set the prior distribution to be a non-informative one.
%Figure~\ref{fig:priordist} shows the non-informative distribution and a sample prior distribution induced from historical data.
In the following, for simplicity, we adopt the non-informative prior distribution $f(\theta) \equiv 1$ where $\theta \in [0,1]$. %We leave as future work the case of having a general prior distribution. We however notice that most of the results in this following remain valid with a general prior distribution. %We will postpone this generalization work to the future and only list some of our known results as remarks.

According to the Bayesian law, the \emph{post distribution} of $\theta$ after witnessing $m$ negative samples and $n = N - m$ positive samples is defined as follows.
\begin{align*}
f(\theta \mid n,m) = \frac{f(\theta)f(n,m\mid\theta)}{\displaystyle\int_{0}^{1} f(\theta)f(n,m\mid \theta)\mathrm{d}\theta} = \frac{ \theta^m (1-\theta)^n}{ \Beta (m+1,n+1)}
\end{align*}
where $\Beta(\ \cdot\ ,\ \cdot\ )$ is the Beta function~\cite{Abramowitz1974}.
The \emph{confidence} we have about the claim that $\theta < \delta$, denoted as $c(n,m,\delta)$, can be defined naturally as the probability of $\theta < \delta$ conditioned on observing the negative and positive samples. Formally,
\begin{align*}
c(n,m,\delta) = \int_0^\delta f(\theta \mid n,m) \mathrm{d}\theta  = \frac{\Beta(\delta;m+1,n+1)}{\Beta (m+1,n+1)}
\end{align*}
where $\Beta(\delta\ ;\  \cdot\ ,\ \cdot \ )$ is the incomplete Beta function~\cite{Abramowitz1974}. 

The following proposition shows that our definition of confidence is consistent with our intuition, i.e., the more positive samples we observe, the more confidence we have.
\begin{proposition}\label{prop:conf2one}
	For any tolerance $0<\delta<1$, $c(n,0,\delta) \to 1$ as $n \to \infty$; and for any $m > 0$,
	$c(n,m,\delta) \to 1$ as $n/m \to \infty$.
\end{proposition}
Thus, if we have dominantly sufficient positive samples, we would always be able to reach a target confidence level. In practice, however, we are always limited in the budget or time, we thus would like to reach a certain confidence level at a low cost. For instance, instead of random sampling, we can apply idea like importance sampling~\cite{Veach1995} so as to increase the probability of sampling a negative sample and hope to gain the same confidence level with fewer samples. Recall that we can view symbolic execution as a particular way of importance sampling. Compared with random sampling, it essentially alters the probabilistic distribution of the traces so that more probability is associated with those traces following a given path. In the following, we investigate the idea of importance sampling in our setting and establish a condition which must be satisfied so that importance sampling (including symbolic execution) must satisfy in order to be more effective in achieving the same confidence level.

\subsection{Importance sampling}\label{sec:importancesampling}
Importance sampling is a widely-used technique in Monte Carlo method in order to approximate the expectation of a probability distribution. The intuition is after observing many positive samples, we should have more confidence in the system's correctness, \emph{if the samples are generated by a method that is more likely to sample a negative one}. We remark that the notion of importance sampling we adopt here has nothing to do with the expectation approximation~\cite{Veach1995}, but rather shares the same idea with the importance sampling in the Monte Carlo method.

Recall that $\theta$ is the error probability. The probability of a specific sampling method finding a negative trace is a function of $\theta$, denoted as $\varphi(\theta)$. We refer to $\varphi(\theta)$ as the \emph{effectiveness function} of the sampling method. Furthermore, $\varphi(\theta)$ is assumed to be continuous and strictly increasing on $[0,1]$ with $\varphi(0)=0$ and $\varphi(1)=1$. Given a specific sampling method (e.g., random sampling or symbolic execution), we may be able to approximate its effectiveness through empirical study. In certain special cases, we might identify a closed form of the effectiveness function for a specific sampling method. For instance, in the case of random sampling, the effectiveness function $\varphi(\theta) = \theta$. A sampling method with effectiveness $\varphi(\theta)$ is said to be \emph{more effective} than another with effectiveness $\phi(\theta)$, if $\varphi(\theta) > \phi(\theta)$ for all $\theta$. Two sampling methods are called incomparable if no one is more effective than the other.
%%%sunjun: check if this is what you mean
%while in the case of a more advanced method, the effectiveness function $\varphi(\theta)$ is likely to be greater than or equal to $\theta$ pointwisely.

In the following, we show that a more effective sampling method leads to a higher confidence level. %We briefly discuss the case where two sampling methods are incomparable at the end of this subsection.
Without loss of generality, we focus on effectiveness functions which can be expressed in the form of a power function $\varphi(\theta) = \theta^\alpha$ where $\theta \in [0,1]$ for $0 < \alpha \le 1$. The reason for the assumption is that effectiveness functions in this form can be compared easily. %Figure~\ref{fig:eff_func} visualizes how the effectiveness function varies with different value for $\alpha$.

Following the discussion in Section~\ref{sec:beyinf}, suppose that the effectiveness of a testing method is $\varphi(\theta) = \theta^\alpha$ and we have witnessed $m$ negative samples and $n = N - m$ positive samples. The post distribution can be calculated as follows.
\begin{align}\label{eq:postdist2}
f(\theta \mid n,m) = \frac{ \theta^{\alpha m} (1-\theta^\alpha)^n}{\displaystyle\int_{0}^{1}  \theta^{\alpha m} (1-\theta^\alpha)^n\mathrm{d}\theta}
\end{align}
Accordingly, the confidence is defined as follows.
\begin{align}\label{eq:conf2}
c_\varphi(n,m,\delta) &= \int_0^\delta f(\theta \mid n,m) \mathrm{d}\theta
=\frac{\displaystyle \int_0^\delta \theta^{\alpha m} (1-\theta^\alpha)^n \mathrm{d}\theta}{\displaystyle\int_{0}^{1}  \theta^{\alpha m} (1-\theta^\alpha)^n\mathrm{d}\theta}.
\end{align}
\iffalse
It can be proved that $\int_0^\delta \theta^{\alpha m} (1-\theta^\alpha)^n \mathrm{d}\theta
	= \frac{1}{\alpha} \Beta(\delta^\alpha; m + \frac{1}{\alpha}, n+1)$ for all $n,m \in \naturals$; $0 < \alpha \le 1$, and $0 < \delta \le 1$. Thus, we simplify equations~\eqref{eq:conf2} and obtain the following result.
\begin{equation}\label{eq:conf}
c_\varphi(n,m,\delta) = \frac{\Beta(\delta^\alpha;m+1/\alpha,n+1)}{\Beta(m+1/\alpha,n+1)}.
\end{equation}
\fi
The following theorem then establishes that a \emph{more effective} sampling method would always result in more confidence.
\begin{theorem}\label{thm:moreconf}
	Let $\varphi(\theta) = \theta^\alpha$ and $\psi(\theta) = \theta^\beta$ be the effectiveness function of two testing methods. If $1 \ge \alpha > \beta > 0$, then $\varphi(\theta) \le \psi(\theta)$ for all $\theta \in [0,1]$, and $c_\varphi(n,m,\delta) \le c_{\psi}(n,m,\delta)$ for all $n, m \in \naturals$ and $\delta \in (0,1)$. \hfill \qed
\end{theorem}
The (rather involved) proof is presented in Appendix A. %Figure~\ref{fig:confs} illustrates how effectiveness influences confidence, i.e., the more effectiveness, the more confidence. %As we can tell from the figures, for any effectiveness rate $\alpha$, the more positive cases we see (i.e., $n \to \infty$) and the less negative cases we see (i.e., $m \to 0$), the more confidence we have (i.e., $c \to 100\%$). The higher the effectiveness rate of a testing method, the more confidence we have about the correctness of systems, giving $n$ positive cases and $m$ negative cases.
This theorem endorses our intuition that we should have more confidence in the systems' correctness when observing many positive samples, if the samples are generated by a method like symbolic execution (with a bias on negative samples). In general we cannot compare the confidence of two incomparable sampling methods. We remark that this result has not been formally proved before and it serves the foundation for the approach we propose next.

Based on the theorem, in order to apply symbolic execution to achieve better confidence than applying random sampling, we should apply it such that it is more likely to sample a negative trace. There are multiple different strategies on how/when to apply symbolic execution. For instance, we could symbolically execute a path which ends with a negative mode, or a part of the path (e.g., we solve for input values which are required to trigger the first few transitions of a path leading to a negative mode, if we have reasons to believe that only those transitions are unlikely to be fired through random sampling), or even simply symbolically execute a path which has not been visited before if all existing samples are positive. In the next section, we discuss how to compare these different strategies based on cost and propose a cost-effect algorithm.

\section{Sampling Strategies}\label{sec:strategy}
Recall that our objective is to check whether there is a trace which visits a negative mode. Theorem~\ref{thm:moreconf} certainly does not imply we should abandon random sampling. The simple reason is that it ignores the issue of time cost. In general, the cost of obtaining a negative trace through sampling (either random sampling or symbolic execution) is: $c/pr$ where $c$ is the cost of obtaining one sample and $pr$ is the probability of the sample being a negative trace. In the case of random sampling, $c$ is often low and $pr$ is also likely low, especially so if the negative traces all contain certain rare event. In the case of symbolic execution, $c$ is likely high since we need to solve a path condition to obtain a sample, whereas $pr$ is likely high. %We remark that the issue of time cost is particularly relevant for hybrid systems since solving a path condition involving ODEs is often very time consuming.
Thus, in order to choose between random sampling and symbolic execution, we would like to know their time cost, i.e., $c$ and $pr$. While knowing the cost of obtaining one sample through random sampling is relatively straightforward, knowing the cost $c$ of symbolic execution is highly non-trivial. In this section, we assume there is a way of estimating that and propose an algorithm based on the assumption. In Section~\ref{experiments}, we estimate $c$ empirically and show that even a rough estimation serves a good basis for choosing the right strategy. We can calculate $pr$ based on ${\cal M}_{\cal H}$. However, constructing ${\cal M}_{\cal H}$ is infeasible and thus we propose to approximate ${\cal M}_{\cal H}$ at runtime. %In the following, we discuss how we estimate $pr$ in detail.

%We have discussed in Section~\ref{sec:importancesampling} the effectiveness of testing methods for discrete-time Markov chains at a high level. The broad picture that we presented is essential to understand the motivation of our method; while in this section, we will study the \emph{efficiency} of methods with different testing strategy.
%In Section~\ref{sec:HA&DTMC} we have shown how hybrid automata as a formal model of hybrid systems can be abstracted by discrete-time Markov chains. From this point of view, a testing method is a way of systematically exploring the Markov chains.
\subsection{Probability Estimation}
Initially, since we have no idea on the probability of obtaining a negative trace, we apply wishful thinking and start with random sampling, hoping that a negative trace will be sampled. If a negative trace is indeed sampled, we are done. Otherwise, the traces that have been sampled effectively identify a subgraph of ${\cal M}_{\cal H}$, denoted as ${\cal M}_{sub}$, which contains only modes and transitions visited by the sample traces. Without any clue on the transition probability between modes not in ${\cal M}_{sub}$, it is infeasible for us to estimate $pr$ (i.e., the probability of reaching a negative mode). It is clear however that in order to reach a negative mode, we must sample in a way such that ${\cal M}_{sub}$ is expanded with unvisited modes. Thus, in the following, we focus on finding a strategy which is cost-effective in discovering new modes instead. 

For each mode $u$ in ${\cal M}_{sub}$ where there is an unvisited child mode $v$, we have the choice of either trying to reach $v$ from $u$ through more random sampling or through symbolic execution (i.e., solving the path condition). In theory, the choice is to be resolved as follows: random sampling if $c_t/q(u) < c_s$ and symbolic execution if $c_t/q(u) \ge c_s$, where $c_t$ is the cost of generating a random sample; $c_s$ is the cost of applying symbolic execution to generate a sample visiting an unvisited child of $u$ in ${\cal M}_{sub}$; and $q(u)$ is the probability of finding a new mode from $u$, i.e., $q(u) = \sum_{v \in V \setminus V_0} p_{uv}$. Intuitively, for random sampling, the expected number of samples to find a new mode is $1/q(u)$ and thus the expected cost of using random sampling to discover a new mode is $c_t / q(u)$. Unfortunately, knowing $q(u)$ and ${\cal M}_{sub}$ exactly is expensive. The former is the subject of recent research on model counting and probabilistic symbolic execution~\cite{DBLP:conf/spin/FilieriFPV15,DBLP:conf/kbse/LuckowPDFV14,DBLP:conf/sigsoft/FilieriPVG14,DBLP:conf/tacas/ChistikovDM15}, and the latter has been studied in~\cite{DBLP:conf/cade/AzizWD12}. Thus, in this work, we develop techniques to estimate their values.

In our approach, we actively estimate the probability of $q(u)$ (for each $u$ in ${\cal M}_{sub}$) from historical observation through Bayesian inference, by recording how many times we sampled $u$. Assume $q = q(u)$ has a prior distribution $f(q)$.
%and Then under the condition that $q = q_0$, the probability of $v$ is an unvisited point is $q_0$, and the probability of $v$ is a visited one is $1 - q_0$.
Let $A$ denote the event that an unvisited child $v$ remains unvisited after one sampling, and $\bar{A}$ denote the event that $v$ becomes visited afterwards. Then,
\begin{equation*}
f(q \mid A) = \frac{f(q)f(A\mid q)}{\displaystyle \int_0^1 f(q)f(A\mid q) \mathrm{d}q}
=\frac{qf(q)}{\displaystyle \int_0^1 qf(q) \mathrm{d}q}
\propto qf(q),
\end{equation*}
and similarly: $f(q \mid \bar{A}) \propto (1-q)f(q)$.

%\begin{figure}
%	\centering
%	\includegraphics[width=0.4\textwidth]{image/comparison.png}
%	\caption{Comparison of the three PDFs.}
%	\label{fig:comparison}
%\end{figure}
Suppose that mode $u$ has been sampled for $N$ times and for $m$ out of $N$ times, we end up with a child which has been visited previously. As a result, $n = N - m$ is the number of times we ended up with an unvisited child. We can compute the post distribution $f(q) \propto (1-q)^mq^n$ and the expectation as follows.
\begin{align}
E(q)
&= \frac{\displaystyle \int_0^1 q (1-q)^m q^n \mathrm{d}q}{\displaystyle \int_0^1 (1-q)^m q^n \mathrm{d}q} \nonumber =\frac{n+1}{m+n+2}.
\label{eq:Eq}
\end{align}
Thus, we estimate $q(u)$ as $(n+1)/(m+n+2)$. Intuitively, the bigger $m$ is, the less likely that an unvisited mode is going to be sampled through random sampling.

Next, we discuss how to apply symbolic execution in this setting. There are multiple strategies on how to apply symbolic execution to construct a sample for visiting $v$. For instance, we could solve a path condition from an initial mode to $v$ so that it will surely result in a trace visiting $v$ (if the path condition is satisfiable). Alternatively, we could take a sample trace which visits $u$ and apply symbolic execution to see whether the trace can be altered to visit $v$ after visiting $u$ by letting a different amount of time elapsing at mode $u$. That is, assume that $(u, X)$ is a concrete state of $sem({\cal H})$ visited by a sample trace, where $X$ is a valuation of $V$. We take the state $(u, X)$ as the starting point and apply symbolic execution to solve a one-step path condition so that $v$ is visited from state $(u, X)$. This is meaningful for hybrid automata because, for every step, by letting a different amount of time elapsing, we may result in firing a different transition and therefore reaching a different mode. We remark that if solving the one-step path condition has no solution, it does not necessarily mean that $v$ is unreachable from $u$. Nonetheless, we argue that this strategy is justified as, according to Theorem~\ref{thm:moreconf}, such a sampling strategy would be more effective than random sampling. To distinguish these two strategies, we refer to the former as \emph{global concolic sampling} and the latter \emph{local concolic sampling}. %In addition, we compare with two alternatives approaches. One is random sampling and the other is applying random testing once and applying symbolic execution to visit the alternative path in the last branch and so on, following~\cite{sen2007concolic}.
The choice of strategy depends on $c_s$. We estimate $c_s$ for particular solvers and systems in Section~\ref{experiments} and choose the right strategy accordingly.

\subsection{Concolic Sampling}\label{sec:algorithms}
Based on the theoretical discussion presented above, we then present our sampling algorithm, which we call concolic sampling. The details are shown in Algorithm~\ref{alg:full}. The input is a hybrid automaton modeling a hybrid system, where some modes are identified as negative ones. The aim is to identify a trace which visits a negative mode or otherwise report that there is certain probabilistic guarantee on their absence. We remark that we skip the part on how the probabilistic guarantee is computed and refer the readers to~\cite{benjaminsmc} for details. We rather focus on our contribution on combining random sampling and symbolic execution for better counterexample identification in the following.

We maintain the set of sample traces as $\Pi$ in the algorithm. Based on $\Pi$, we can construct the above-mentioned subgraph ${\cal M}_{sub}$ of ${\cal M}_{\cal H}$ systematically. Next, for each node $u$ in ${\cal M}_{sub}$ which potentially has unvisited children, we maintain two numbers $m$ and $n$ as discussed above, in order to estimate the probability of $q(u)$. If according to our strategy, there is still some $u$ such that it might be cheaper to discover a new child mode through random sampling, we proceed by generating a random sample using the algorithm presented in~\cite{benjaminsmc}, which is shown as Algorithm~\ref{alg:random}.

\begin{algorithm}[t]
    Let $\Pi$ be the set of sampled runs, initialized to the empty set; \\
    Let ${\cal M}_{sub}$ be the subgraph of $\cal M$, initialized to the root node; \\
    \Repeat {time out} {
       	Set $u = \argmin_u \min\left(c_t/E(q(u)), c_s\right)$; \\
        \If {$c_t/q(u) < c_s$} {
			Invoke Random Sampling to generate a run $\pi$;
		}
        \Else {
            Apply symbolic execution to obtain a sample $\pi$ visiting a new child of $u$; \\
%            Simulate ${\cal H}$ according to the sample to get a run $\pi$; \\
        }
		\If {$\pi$ visits a negative mode} {
			Present $\pi$ as a counterexample and terminate;
		}
		\If {$\pi$ visits an undiscovered mode} {
			$n_u \coloneqq n_u + 1$;
		}
		\Else {
			$m_u \coloneqq m_u + 1$;
		}
		Add $\pi$ into $\Pi$ and add $u$ to $\MM_{sub}$;
	}
	\caption{Concolic Sampling}
	\label{alg:full}
\end{algorithm}

We briefly introduce how Algorithm~\ref{alg:random} works in the following. In a nutshell, the algorithm is designed to sample a run $\pi$ according to an approximation of the probability distribution of ${\cal M}_{\cal H}$. The main idea is to use the Monte Carlo method to approximate the measure of time windows $\Vert T_q(v,g) \Vert$. Recall that $T_q(v,g) = \{t \in (0,1) \mid \theta_q(v,t) \text{ satisfies } g\}$. Therefore the measure $\Vert T_q(v,g) \Vert$ is the mean of $[\theta_q(v,\tau) \models g]$, where the random variable $\tau \sim \UU(0,1)$. According to the law of large numbers~\cite{Leon-Garcia:2007:PRP:1214239}, the sample mean almost surely converges to the expectation. Thus we have
\begin{equation*}
\frac{\sum_{i=1}^n [\theta_q(v,\tau_i)\models g]}{n} \xrightarrow{a.s.} \Vert T_q(v,g_j) \Vert
\quad \text{as} \quad n \to \infty,
\end{equation*}
where $\tau_1,\tau_2,\dots,\tau_n \distras{i.i.d.} \tau$. To generate a $K$-step run, Algorithm~\ref{alg:random} works by generating one random step at a time. In particular, after each time unit, at line 4, firstly a set of time points are uniformly generated. By testing how often each transition from the current mode is enabled at these time points, we estimate the transition probability in ${\cal M}_{\cal H}$. At line 7, we sample a transition according to the estimated probability and generate a step. This procedure finishes when a run which spans $K$ time units is generated.

If random sampling is unlikely to be cost-effective in discovering a new mode according to our strategy, symbolic execution is employed at line 7 in Algorithm~\ref{alg:full} to generate a sample to cover a new node in ${\cal M}_{\cal H}$. Among all the nodes in ${\cal M}_{sub}$, we identify the one which would require the minimum cost to discover a new child according to our estimation $c_s$, encode the corresponding path condition and apply an existing constraint solver that supports ODE (i.e., dReal~\cite{gao2013dreal}) to generate a corresponding run. Once we obtain a new run $\pi$ at line 8, we check whether it is a counterexample. If it is, we report and terminate at line 10. Otherwise, we repeat the same procedure until it times out. %We leave the details on the choices of how to apply symbolic execution in Section~\ref{experiments} after presenting the details on estimating $c_s(u, {\cal M}_{sub})$. \\

\begin{algorithm}[t]
	\KwIn{A hybrid automaton ${\cal H}$ and a state $\langle q,v \rangle$}
	\KwResult{A successive state $\langle p,u \rangle$}
			Sample time points $t_1,\cdots,t_J$ uniformly from $[0,1]$; \\
			\ForEach {outgoing transition $q \xrightarrow{g_i} q_i$}  {
				Set $T_i:= \{ t_j \mid \Phi_{q}(t_j,v) \models g_i \}$,
			}
			Choose a transition $q \xrightarrow{g_i} q_i$ with probability $\Vert T_i \Vert /\sum_i \Vert T_i \Vert$; \\
			Sample a time point $t$ uniformly from $T_i$; \\
			Set $u \coloneqq \Phi_{q_i}(1-t, \Phi_{q}(t,v))$ and $p \coloneqq q_i$; \\
		Return $\langle p,u \rangle$;
	\caption{Random Sampling}
	\label{alg:random}
\end{algorithm}

%\begin{algorithm}
%	\KwIn{TBC}
%	\KwResult{TBC}
%	$u \coloneqq root$\\
%	\Repeat {$N$ times} {
%		Pick a random initial valuation $v_0$; \\
%		Set $\pi$ to be $\langle (q_0,v_0) \rangle$ where $q_0 = q_{in}$; set $k$ to be 0; \\
%		\While {$k \leq K$} {
%			Sample a set of $J$ time points uniformly from $[0,\Delta]$; \\
%			\For {each outgoing transition $t_i = (q_k, g_i, q_i)$ from $q_k$} {
%				Set $T_i:= \{ t \in J \mid \Phi_{q_k}(t,v_k) \in g_i \}$,
%			}
%			Choose a transition $t_i$ with probability $|T_i| / \left(\sum_i |T_i| \right)$; \\
%			Sample a time point $t$ uniformly from $T_i$; \\
%			Set $v_{k+1} := \Phi_{q_{k+1}}(1-t, \Phi_{q_k}(t,v_k))$; set $q_{k+1} = q_i$; \\
%			Set $\pi := \pi \cdot \langle t_i, (q_k, v_k) \rangle$; \\
%		}
%		\If {$\pi$ fails $\varphi$} {
%			present $\pi$ as a counterexample and terminate;
%		}
%		Add $\pi$ into $\Pi$;
%	}
%	Return $\Pi$
%	\caption{HyChecker (light)}
%	\label{alg:light}
%\end{algorithm}

\section{Evaluation} \label{experiments}
We implemented our approach in a self-contained toolkit called \HC, available online at~\cite{url}.
\HC{} is implemented with 1575 lines of Python codes (excluding external libraries we used) and is built with a web interface. \HC{} relies on the dReal constraint solver~\cite{gao2013dreal} for symbolic execution. In the following, we evaluate \HC{} in order to answer the following research question: does our strategy on combining random sampling and symbolic execution (resulting from our theoretical analysis) allow us to identify rare counterexamples more efficiently?

Our test subjects include three benchmark hybrid systems which we gather from previous publications as well as a simplified real-world water treatment system.
\begin{itemize}
    \item \emph{Thermodynamic system} We first test our method on a room heating system from~\cite{fehnker2004benchmarks}. The system has $n$ rooms and $m \le n$ heaters which are used to tune the rooms' temperature. % within a comfort zone $(T_l,T_h)$.
The temperature of a room is affected by the environment temperature and also by whether a heater is warming the room. The system therefore aims to maintain the rooms' temperature within certain comfortable range by moving around and turning on and off the heaters. %\footnote{We refer to~\cite{fehnker2004benchmarks} for further details about the control logics and parameters values of the thermodynamic system.}.
%, the temperature difference with the adjacent rooms and whether there is a heater inside. The heater is turned off if the temperature exceeds the upper bound $high\_temp$ or turned on if the temperature falls beyond the lower bound $low\_temp$. If the following condition is fulfilled: the temperature in current room is smaller than $get\_heater$, the temperature difference with the adjacent room is larger than $dif\_temp$, the current room does not have a heater and the adjacent room has a heater, the heater in the adjacent room is transferred to the current room and set to be on.
We consider in the experiment such a system with three rooms and two heaters. %Let $\vec{x}$ be the temperature in the three rooms, constant $u = 25$ be the environment temperature, and a boolean vector $\vec{h} = (h_1,h_2,h_3)^T$ indicating whether a working heater is in a room. The dynamic behavior is then modeled as
%\begin{equation*}
%\frac{\mathrm{d}\vec{x}}{\mathrm{d}t}
%=
%\begin{pmatrix}
%-0.9	&	0.5		&	0	\\
%0.5		&	-1.3	&	0.5	\\
%0		&	0.5		&	-0.9
%\end{pmatrix} \vec{x}
%+
%\begin{pmatrix}
%0.4 u + 6 h_1 \\
%0.3 u + 7 h_2\\
%0.4 u + 8 h_3
%\end{pmatrix},
%\end{equation*}
%where initially $\vec{x}_0 \in [20,20.5]^3$ and $h_1 = h_2 = true$.
We verify the same property as in~\cite{benjaminsmc}, i.e., whether the two heaters will be moved to other rooms in the first five days.

\item \emph{Navigation system} Our second test subject is the navigation system from~\cite{fehnker2004benchmarks}. This system contains a grid of cells, where each cell is associated with some particular velocity.
Whenever a floating object moves from one cell to the other, it changes its acceleration rate according to the velocity in that cell. If the object leaves the grid, the velocity is the one of the closest cell.
We check whether an object in the grid will leave its initial cell and will not enter a dangerous cell, within six minutes.% The corresponding BLTL $\varphi = \lozenge^{\le 240} \lnot \mathsf{Init} \land \square^{\le 240} \lnot \mathsf{Dgrs}$.

\item \emph{Traffic system} Our third model is from the long standing research on modeling traffic and examining causes of traffic jams and car crashes. We use the ODE in~\cite{PhysRevE.80.046205} to describe the dynamics of a vehicle. We consider in the experiment a circular road with $n = 5$ cars on it. We are interested whether there could be a potential traffic accident in the closed system, and whether there could be a potential traffic jam. %An accident happens if $h_i(t) < 0$ for some $i$, and a jamiton occurs if $v_i(t) < v_{\text{slow}}$ for some $i$.

\item \emph{SWaT system}
Lastly, we tested our method on a simplified real-world system model.
The Secure Water Testbed (SWaT) is a raw water purification laboratory located at SUTD~\cite{swat}. SWaT is a complicated system involving a series of water treatments like ultrafiltration, chemical dosing, dechlorination through an ultraviolet system.
We build a hybrid automaton model of SWaT based on the control program in the programmable logic control (PLC) in the system. The modes are defined based on the discrete states of the actuators (e.g., motorized valves and motorized pumps). These actuators are controlled by the PLC. There are in total 23 actuators, which results in many modes.
By focusing on the hydraulic process in the system only, we build a hybrid automaton with 2721 modes. We skip the details of the model due to the limited space here. The readers are referred to~\cite{hychecker} for details.
The property we verify is that the water level in the backwash tank must not be too high or too low (otherwise, the system needs to be shut down), with some extreme initial setting (e.g., the water level in the tank is close to be too low) to analyze the system safety.
\end{itemize}

\paragraph{Estimating Cost of Symbolic Execution} In order to apply Algorithm~\ref{alg:full} with the right strategy, we need to estimate $c_s$. The underlying question is how efficient a constraint solver can check the satisfiability of a given path condition. We remark that it is a challenging research problem and perhaps deserves a separate research project by itself. There are a dozen of various factors that determines how a solver performs in solving a given constraint, including the number of variables, the number of operators, the number of differential equations, the length of witnesses (if there is any), etc. Even on the same problem, different solvers have different performance due to different search strategies, ways of pruning and reducing state space, etc~\cite{Lu:2005:ESS:1048925.1049282,Armand2011,Jha2009}. All these facts make a precise estimation of efficiency of symbolic solvers extremely difficult.

In this work, following previous work on this topic~\cite{DBLP:conf/cade/AzizWD12}, we estimate $c_s$ as follows. First, we construct a sequence of increasingly more complicated random constraints (composed of constraints on ODE as well as ordinary constraints, which we obtain from examples in dReal). We then measure the time needed to solve them using dReal one-by-one. Based on the results, we observe that the dominant factor is the length of the formula and thus heuristically decide $c_s$ to be a function of the length of constraints. Next, we apply a function fitting method to obtain a function for predicting $c_s$. The function we obtained is $\exp(1.73 l - 1.65)-1$ where $l$ is the length of the formula, which suggests that the solving time is exponential in the length of the formula. It implies that dReal has problem solving path conditions containing two or more steps, which in turn means that our choice of strategy should be the \emph{local concolic sampling}. We remark that this is unlikely a precise estimation. Nonetheless, as we show below, even a rough estimation would be useful in guiding when and how to do symbolic execution.

\paragraph{Experiment Results}
Table~\ref{roomheater} shows the experiment results. All experiment results are obtained in Ubuntu Linux 14.04 on a machine with an Intel(R) Core(TM) i5-4950, running with one 3.30GHz CPU core (no parallel optimization), 6M cache and 12 GB RAM. We set a timeout of 10 minutes for each experiment, i.e., if no counterexample is identified after 10 minutes, the property has passed the test. Each experiment is repeated for 10 times and we report the average time. All details on the experiments are at~\cite{hychecker}. %If the verification results vary, we report how often each result is obtained.

We compare four approaches in order to show the effectiveness of our chosen strategy. The first is the random sampling approach proposed in~\cite{benjaminsmc}. The results are shown in column {\bf random}. Note that there are two results for the thermodynamic system. This is because due to the randomness in the approach, the results are not always consistent (e.g., in one experiment, a counterexample is found, whereas none is found in another). The second approach is the concolic testing approach in~\cite{sen2007concolic} (i.e., applying random testing once and applying symbolic execution to visit the alternative path in the last branch and so on). The results are shown in column {\bf dynamic}. The last two columns report the result of applying {\bf global} concolic sampling and {\bf local} concolic sampling respectively. %Note that in each case the last row shows how many runs are sampled.

We have the following observations based on the results. First, among the four approaches, local concolic testing is able to spot counterexamples more efficiently in all cases. Compared with random sampling, the number of samples explored by local concolic sampling is significantly smaller. This confirms the result of the theoretical analysis in previous sections. Second, symbolic execution for hybrid systems are clearly constrained by the limited capability of existing hybrid constraint solvers like dReal. For all four cases, both concolic testing and global concolic sampling time out whilst waiting for dReal to solve the first path condition. This is because the path condition (composed of constraints from multiple steps) is complex and dReal takes a lot of time trying to solve it. The only difference is that while concolic testing got stuck after the first sample, global concolic sampling got stuck after it has randomly sampled a few traces and switched to symbolic execution. On the contrary, local concolic sampling uses dReal to solve a one-step path condition each time and is able to smartly switch between random sampling and symbolic execution, and eventually found a counterexample. Third, the experiment results suggest that the formula that we applied for estimating $c_s$ turned out to be an under-approximation, i.e., the actual time cost is often much larger. If we modify the function to return a much larger cost for solving a path composed of two or more steps, global concolic sampling would be equivalent to random sampling as symbolic execution would never be selected due to its high cost.

\begin{table}[t]
\centering
\caption{Experiments results}
\label{roomheater}
\begin{tabular}{|c|c|c|c|c|c|c|}
\hline
& & \multicolumn{2}{c|}{\textbf{random}} & \textbf{dynamic} & \textbf{global} & \textbf{local} (our strategy) \\
\hline
%\multicolumn{6}{|c|}{\emph{thermodynamic system}} \\
%\hline
\multirow{3}{*}{\emph{thermodynamic system~}} & result &ct-eg found&pass & pass & pass & ct-eg found \\
& time(s) &$340.4$ & $600$& $600$ & $ 600$ &$41.25$ \\
& \#samples &$13$K &$22$K & n/a & $134$ &$551$ \\
\hline
%\multicolumn{6}{|c|}{\emph{navigation system}} \\
%\hline
\multirow{3}{*}{\emph{navigation system~}} & result &\multicolumn{2}{c|}{ct-eg found} & pass & pass &ct-eg found \\
& time(s) &\multicolumn{2}{c|}{$91.7$ } & $ 600$& $ 600$ & $4.33$  \\
& \#samples &\multicolumn{2}{c|}{$354$ } & n/a & $4$ & $5$ \\
\hline
%\multicolumn{6}{|c|}{\emph{traffic system}} \\
%			\hline
\multirow{3}{*}{\emph{traffic system~}} & result &\multicolumn{2}{c|}{pass } &pass &pass &ct-eg found\\
	&		time(s) &\multicolumn{2}{c|}{$600$ } &$600$ &$600$ &$28.86$ \\
	&		\#samples &\multicolumn{2}{c|}{$1240$ }  &n/a &$2$ &$2$ \\
    			\hline
%\multicolumn{6}{c|}{\emph{SWaT system}} \\
%    \hline
\multirow{3}{*}{\emph{SWaT system~}} & result &\multicolumn{2}{c|}{ct-eg found}   & pass&pass &ct-eg found \\
& time(s) &\multicolumn{2}{c|}{$102.4$} &$600$ &$600$ &$64.6$ \\
& \#samples &\multicolumn{2}{c|}{$169$ } &n/a &$24$ &$68$ \\
\hline
\end{tabular}
\end{table}

%
%\noindent \emph{Thread to Validity} First, we acknowledge that the subjects used for evaluation might be biased. Though the models and the properties are randomly selected, they may not be representative of other models. Second, we admit the estimation of $c_s$ in this work can be improved. Though it may be that we could never construct an accurate function for estimating $c_s$, a rough estimation like the one we used has been shown to be useful in guiding the selection of different sampling methods.

\section{Conclusion and Related Works}\label{sec:relatedworks}
In this work, we investigated the effectiveness of different sampling methods (i.e., random sampling and symbolic execution) for hybrid systems. We established theoretical results on comparing their effectiveness and we developed an approach for combining random sampling and symbolic executions in a way which is provably cost-effective.

In the following, we discuss the related work, in addition to those discussed already.
%This work is inspired by~\cite{benjaminsmc}. In~\cite{benjaminsmc}, the authors proposed to verify hybrid automata models probabilistically by abstracting hybrid automata as DTMC. We extend their idea by developing our sampling strategy based on the DTMC abstraction and formally establish its effectiveness.
This work is inspired by~\cite{DBLP:conf/sigsoft/BohmeP14}, which initialized the discussion on the efficiency of random testing.
Our work aims to combine random sampling and symbolic execution to identify rare counterexamples efficiently. It is thus closely related to work on handling rare events in the setting of statistical model checking~\cite{barbot2012coupling,barbot2012importance,DBLP:conf/cav/JegourelLS13}. In~\cite{barbot2012coupling}, the authors set up a theoretical framework using coupling theory and developed an efficient sampling method that guarantees a variance reduction and provides a confidence interval. In~\cite{barbot2012importance} the authors proposed the first importance sampling method for CTMC to provide a true confidence interval.  In~\cite{DBLP:conf/cav/JegourelLS13} the authors motivated the use of importance splitting to estimate the probability of a rare property. Our work is different as we complement sampling with symbolic execution to identify rare events efficiently.

%For dealing with rare events, a majority of two solutions were adopted. One is importance sampling, the other is importance splitting.
%
% In \cite{barbot2012coupling}, the rare event is considered as the major drawback of statistical model checking, as traditional statistical ways' time consumption will increase by magnitudes. In this work, a structural analysis of the input model is required. And no numerical computation is need because the usage of coupling theory.
%
% \cite{barbot2012importance} works on timed temporal formula on continuous time Markov Chains. It is able to generate true confidence interval on events with very small probability. It has also proposed a framework on reduction of the variance.
%
% In \cite{jegourel2013importance}, the authors derives a score function for both fixed and adaptive level splitting for rare events. Three case studies are given in this work while no one is related to ordinary differential equations.

This work borrows idea from work on combining program testing with symbolic execution (a.k.a.~dynamic symbolic execution or concolic testing). In~\cite{Godefroid:2005:DDA:1064978.1065036}, the authors proposed a way of combining program testing with symbolic execution to achieve better test coverage. Random testing is first applied to explore program behaviors, after which symbolic execution is used to direct the test towards different program branches. Similar ideas later have been developed in~\cite{sen2006cute,majumdar2007hybrid,sen2007concolic,DBLP:conf/osdi/CadarDE08}. Our work is different in two ways. One is that we target hybrid systems in work, which has different characteristics from ordinary programs. One of them is that symbolic execution of hybrid automata is considerable more expensive, which motivated us to find ways of justifying the use of symbolic execution. The other is that, based on the probabilistic abstraction of hybrid models, we are able to formally compare the cost of random sampling against symbolic execution to develop cost-effective sampling strategies. We remark that the same idea can be applied to concolic testing of programs as well.

\HC{} is a tool for analyzing hybrid systems and thus it is related to tools/systems on analyzing hybrid systems. In~\cite{platzer2010logical}, the authors developed a theorem prover for hybrid systems. Users are required to use differential dynamic logic to model hybrid systems. Afterwards, the prover can be used interactively to find a sound and complete proof of certain properties of the system. It has been shown that the prover works for safety critical systems like aircrafts~\cite{platzer2010logical}. \HC{} is different as it is fully automatic. \emph{dReach}~\cite{gao2014delta} is a recent tool developed for verifying hybrid systems. It is based on the SMT solver \emph{dReal}~\cite{gao2013dreal} developed by the same authors. \emph{dReach} focuses on bounded $\delta$-complete reachability analysis. It provides a relatively easy-to-use interface for modeling hybrid systems and verifies whether a system is $\delta$-safe under given safety demands. We observe since \emph{dReach} attempts to solve every path in a hybrid automaton, its performance suffers when the system becomes more complicated. \HC{} relies on \emph{dReal} and tries to improve \emph{dReach} by combining random sampling to avoid solving many of the paths. HyTech~\cite{henzinger1997hytech} is one of the earliest tools on verifying hybrid systems. It is limited to linear hybrid automata.

\section*{Acknowledgement}
The project is supported by the NRF project IGDSi1305012 in SUTD.\\
The work is supported by the National Natural Science Foundation of China under grant no. 61532019, 61202069 and 61272160. 

\bibliographystyle{abbrv}
\bibliography{paper}

\newpage
\appendix
\section{Proofs}
In this section we prove Proposition~\ref{prop:conf2one} and Theorem~\ref{thm:moreconf}.

\begin{proof}[of Proposition~\ref{prop:conf2one}]
	It is worthwhile to point out that $c(n,m,\delta)$, as a function of $\delta$, is a distribution function~\cite{Crr46} with expectation $\mu$ and variance $\sigma^2$ given by
	\begin{equation*}
	\mu = \frac{m+1}{m+n+2}, \quad \sigma^2 = \frac{(m+1)(n+1)}{(m+n+3)(m+n+2)^2}.
	\end{equation*}
	Obviously, in either limit condition in the proposition we have $\mu \to 0, \sigma^2 \to 0$.
	
	Let $X$ is a random variable according to the distribution function $c(n,m,\cdot\ )$. According to Chebyshev's inequality
	\begin{equation*}
	P(|X-\mu| \le \epsilon) \ge 1- \frac{\sigma^2}{\epsilon^2},
	\end{equation*}
	and let $\epsilon = \delta - \mu$, we have
	
	\begin{equation*}
	P(X \le \delta) \ge P(|X-\mu| \le \delta-\mu) \ge
	1- \frac{\sigma^2}{(\delta - \mu)^2} \to 1.
	\end{equation*}
	
	On the other hand, we have $c(n,m,\delta) = P(X \le \delta) \le 1$. By Squeeze theorem, we conclude $c(n,m,\delta) \to 1$.
\end{proof}

\iffalse
\begin{proof}
	\begin{align*}
	\alpha
	&=	\frac{\displaystyle\int_0^\delta f(\theta) (1-\theta)^n\mathrm{d}\theta}{\displaystyle\int_{0}^{1} f(\theta)(1-\theta)^n\mathrm{d}\theta}\\
	&= \frac{\displaystyle\int_0^\delta f(\theta) (1-\theta)^n\mathrm{d}\theta}{\displaystyle\int_{0}^{\delta} f(\theta)(1-\theta)^n\mathrm{d}\theta + \displaystyle\int_{\delta}^{1} f(\theta)(1-\theta)^n\mathrm{d}\theta},
	\end{align*}
	so
	\begin{align*}
	\frac{1}{\alpha}-1 &= \frac{\displaystyle\int_\delta^1 f(\theta) (1-\theta)^n\mathrm{d}\theta}{\displaystyle\int_{0}^{\delta} f(\theta)(1-\theta)^n\mathrm{d}\theta}
	\le \frac{\displaystyle\int_\delta^1 M (1-\theta)^n\mathrm{d}\theta}{\displaystyle\int_{0}^{\delta} m(1-\theta)^n\mathrm{d}\theta}
	= \frac{M}{m}\cdot \frac{\displaystyle\int_\delta^1 (1-\theta)^n\mathrm{d}\theta}{\displaystyle\int_{0}^{\delta} (1-\theta)^n\mathrm{d}\theta}\\
	&= \frac{M}{m} \cdot \frac{(1-\delta)^{n+1}}{1-(1-\delta)^{n+1}}
	\to 0 \qquad \text{as $n \to \infty$}.
	\end{align*}
	On the other hand, because $\alpha \le 1$, we have
	\[\frac{1}{\alpha}-1 \ge 0.
	\]
	According to the squeeze theorem, we proved that $\alpha \to 1$ as $n \to \infty$.
\end{proof}
\fi

Before we prove Theorem~\ref{thm:moreconf}, We first prove some lemmas. In the following, we assume all functions are Riemann-integrable functions.

\begin{lemma}\label{lem:candy}
	For any $A, B, a, b > 0$,
	\begin{equation*}
	  \frac{A}{B} \le \frac{A+a}{B+b} \iff \frac{A}{B} \le \frac{a}{b}.
	\end{equation*}
\end{lemma}

\begin{proof}
	\begin{align*}
	\frac{A}{B} \le \frac{A+a}{B+b}
	&\iff A(B+b) \le B(A+a) \\
	&\iff AB + Ab \le AB + aB \\
	&\iff Ab \le aB \\
	&\iff 	\frac{A}{B} \le \frac{a}{b}.
	\end{align*}
\end{proof}

\begin{lemma}\label{lem:intcandy}
	For any functions $f, g\colon [0,1] \to \reals_{\ge 0}$ and $\delta \in~(0,1)$,	\begin{equation*}
	\frac{\displaystyle \int_{0}^{\delta} f(\theta)g(\theta)\mathrm{d}\theta}{\displaystyle \int_{0}^{1} f(\theta)g(\theta)\mathrm{d}\theta}
	\ge
	\frac{\displaystyle \int_{0}^{\delta} f(\theta)\mathrm{d}\theta}{\displaystyle \int_{0}^{1} f(\theta)\mathrm{d}\theta},
	\end{equation*}
	if $g$ is a decreasing function.
\end{lemma}

\begin{proof}
	It is sufficient to show
    \begin{align*}
        \frac{\displaystyle \int_{\delta}^{1} f(\theta)g(\theta)\mathrm{d}\theta}{\displaystyle \int_{0}^{\delta} f(\theta)g(\theta)\mathrm{d}\theta}
        \le
        \frac{\displaystyle \int_{\delta}^{1} f(\theta)\mathrm{d}\theta}{\displaystyle \int_{0}^{\delta} f(\theta)\mathrm{d}\theta}.
    \end{align*}
    Since $g$ is a decreasing function,
    \begin{align*}
      \frac{\displaystyle \int_{\delta}^{1} f(\theta)g(\theta)\mathrm{d}\theta}{\displaystyle \int_{0}^{\delta} f(\theta)g(\theta)\mathrm{d}\theta}
      &\le
      \frac{\displaystyle \int_{\delta}^{1} f(\theta)g(\delta)\mathrm{d}\theta}{\displaystyle \int_{0}^{\delta} f(\theta)g(\delta)\mathrm{d}\theta}
      =
      \frac{\displaystyle \int_{\delta}^{1} f(\theta)\mathrm{d}\theta}{\displaystyle \int_{0}^{\delta} f(\theta)\mathrm{d}\theta}.
    \end{align*}
\end{proof}

\begin{lemma}\label{lem:step_t}
	Let $f \colon [0,1] \to \reals_{\ge 0}$ be a function and $\delta \in (0,1)$. Then
	\begin{equation*}
	\frac{\displaystyle \int_0^\delta \theta^\alpha  f(\theta) \mathrm{d}\theta}{\displaystyle \int_0^1 \theta^\alpha  f(\theta) \mathrm{d}\theta}
	\le
	\frac{\displaystyle \int_0^\delta \theta^\beta  f(\theta) \mathrm{d}\theta}{\displaystyle \int_0^1 \theta^\beta  f(\theta) \mathrm{d}\theta},
	\end{equation*}
	if $\ 0 < \beta < \alpha \le 1$.
\end{lemma}
\begin{proof}
	Notice that
	\begin{align*}
	\frac{\displaystyle \int_0^\delta \theta^\beta  f(\theta) \mathrm{d}\theta}{\displaystyle \int_0^1 \theta^\beta  f(\theta) \mathrm{d}\theta}
	&=
	\frac{\displaystyle \int_0^\delta \theta^\alpha f(\theta) \mathrm{d}\theta + \int_0^\delta (\theta^\beta - \theta^\alpha) f(\theta) \mathrm{d}\theta}{\displaystyle \int_0^1 \theta^\alpha f(\theta) \mathrm{d}\theta + \int_0^1 (\theta^\beta - \theta^\alpha) f(\theta) \mathrm{d}\theta}.
	\end{align*}
	By Lemma~\ref{lem:candy}, it is sufficient to show
	\begin{align*}
	\frac{\displaystyle \int_0^\delta (\theta^\beta - \theta^\alpha) f(\theta) \mathrm{d}\theta}{\displaystyle \int_0^1 (\theta^\beta - \theta^\alpha) f(\theta) \mathrm{d}\theta}
	\ge
	\frac{\displaystyle \int_0^\delta \theta^\alpha  f(\theta) \mathrm{d}\theta}{\displaystyle \int_0^1 \theta^\alpha  f(\theta) \mathrm{d}\theta},
	\end{align*}
where the left-hand-side
\begin{align*}
\frac{\displaystyle \int_0^\delta (\theta^\beta - \theta^\alpha) f(\theta) \mathrm{d}\theta}{\displaystyle \int_0^1 (\theta^\beta - \theta^\alpha) f(\theta) \mathrm{d}\theta}
=
\frac{\displaystyle \int_0^\delta \theta^\alpha f(\theta) \left(\theta^{\beta - \alpha} - 1\right) \mathrm{d}\theta}{\displaystyle \int_0^1 \theta^\alpha f(\theta) \left(\theta^{\beta - \alpha} - 1\right) \mathrm{d}\theta}.
\end{align*}
Notice that $\theta^{\beta - \alpha} - 1$ as a function of $\theta$ is positive and decreasing in $[0,1]$. Hence, by Lemma~\ref{lem:intcandy},
\begin{align*}
\frac{\displaystyle \int_0^\delta \theta^\alpha f(\theta) \left(\theta^{\beta - \alpha} - 1\right) \mathrm{d}\theta}{\displaystyle \int_0^1 \theta^\alpha f(\theta) \left(\theta^{\beta - \alpha} - 1\right) \mathrm{d}\theta}
\ge
\frac{\displaystyle \int_0^\delta \theta^\alpha  f(\theta) \mathrm{d}\theta}{\displaystyle \int_0^1 \theta^\alpha  f(\theta) \mathrm{d}\theta}.
\end{align*}
\end{proof}

\begin{lemma}\label{lem:step_1-t}
	Let $f \colon [0,1] \to \reals_{\ge0}$ be a function and $\delta \in (0,1)$. Then
	\begin{align*}
	\frac{\displaystyle \int_{0}^{\delta} \left(1-\theta^\alpha\right)f(\theta)\mathrm{d}\theta}{\displaystyle \int_{0}^{1} \left(1-\theta^\alpha\right)f(\theta)\mathrm{d}\theta}
	\le
	\frac{\displaystyle \int_{0}^{\delta} \left(1-\theta^\beta\right)f(\theta)\mathrm{d}\theta}{\displaystyle \int_{0}^{1} \left(1-\theta^\beta\right)f(\theta)\mathrm{d}\theta},
	\end{align*}
	if $\ 0 < \beta < \alpha \le 1$.
\end{lemma}

\begin{proof}
	Let $\displaystyle g(\theta) = \frac{1-\theta^\beta}{1-\theta^\alpha}$. We first show that $g(\theta)$ is decreasing in $[0,1]$. The derivative
	\begin{align*}
	g'(\theta)
	&=
	\frac{-\beta \theta^{\beta-1}(1-\theta^\alpha) - (1-\theta^\beta) (-\alpha \theta^{\alpha-1})}{\left(1-\theta^\alpha\right)^2}\\
	&=
	\alpha \beta \frac{\theta^{\alpha+\beta+1}}{\left(1-\theta^\alpha\right)^2}\left[\left(\frac{\theta^{-\beta}}{\beta} - \frac{1}{\beta}\right)- \left(\frac{\theta^{-\alpha}}{\alpha} - \frac{1}{\alpha}\right)\right].
	\end{align*}
	Let $\displaystyle h_t(x) = \frac{\theta^{-x}}{x} - \frac{1}{x}$ for $x \in (0,1]$. Notice that
	\begin{align*}
	h_t'(x) = \frac{1 - \theta^{-x}}{x^2} + \frac{-\ln \theta}{x}\theta^{-x} > 0
	\end{align*}
	for any $x \in (0,1]$, so $h_t(x)$ is an increasing function. Therefore,
	\begin{align*}
	g'(\theta) = \alpha \beta \frac{\theta^{\alpha+\beta+1}}{\left(1-\theta^\alpha\right)^2}\left[h_t(\beta) - h_t(\alpha)\right] < 0.
	\end{align*}
	
	Hence, by Lemma~\ref{lem:intcandy},
	\begin{align*}
    \frac{\displaystyle \int_{0}^{\delta} \left(1-\theta^\beta\right)f(\theta)\mathrm{d}\theta}{\displaystyle \int_{0}^{1} \left(1-\theta^\beta\right)f(\theta)\mathrm{d}\theta}
    &=
    \frac{\displaystyle \int_{0}^{\delta} \left(1-\theta^\alpha\right)f(\theta)g(\theta)\mathrm{d}\theta}{\displaystyle \int_{0}^{1} \left(1-\theta^\alpha\right)f(\theta)g(\theta)\mathrm{d}\theta}\\
    &\ge
    \frac{\displaystyle \int_{0}^{\delta} \left(1-\theta^\alpha\right)f(\theta)\mathrm{d}\theta}{\displaystyle \int_{0}^{1} \left(1-\theta^\alpha\right)f(\theta)\mathrm{d}\theta}.
	\end{align*}
\end{proof}

And here follows the last proof in the paper.
\begin{proof}[of Theorem~\ref{thm:moreconf}]
We only need to show that
	\begin{align*}
      \frac{\displaystyle \int_0^\delta \theta^{\alpha m} (1-\theta^\alpha)^n \mathrm{d}\theta}{\displaystyle\int_{0}^{1}  \theta^{\alpha m} (1-\theta^\alpha)^n\mathrm{d}\theta}
      \le
      \frac{\displaystyle \int_0^\delta \theta^{\beta m} (1-\theta^\beta)^n \mathrm{d}\theta}{\displaystyle\int_{0}^{1}  \theta^{\beta m} (1-\theta^\beta)^n\mathrm{d}\theta}
	\end{align*}
	holds for any $m,n \in \naturals$, $\delta \in (0,1)$ and $0 < \beta < \alpha \le 1$, which is obviously true by using Lemma~\ref{lem:step_t} for $m$ times and Lemma~\ref{lem:step_1-t} for $n$ times.
\end{proof}

\end{document}